\documentclass[conference,a4paper]{IEEEtran}
\usepackage{multicol}
\usepackage{graphicx}

\DeclareMathAlphabet{\mathpzc}{OT1}{pzc}{m}{it}
\usepackage{bbm}
\usepackage{dsfont}
\usepackage{yfonts}
\usepackage{amsmath}
\usepackage{amssymb}
\usepackage{amsfonts}
\usepackage{mathrsfs}

\usepackage{algorithm} %format of the algorithm
\usepackage{algorithmic} %format of the algorithm
\usepackage{multirow} %multirow for format of table
\usepackage{xcolor}
\usepackage{cite}

\DeclareMathOperator*{\argmax}{argmax}
\DeclareMathOperator*{\argmin}{argmin}

\usepackage{multicol}
\usepackage{graphicx}

\DeclareMathAlphabet{\mathpzc}{OT1}{pzc}{m}{it}
\usepackage{bbm}
\usepackage{dsfont}
\usepackage{yfonts}
\usepackage{amssymb}
\usepackage{amsfonts}
\usepackage{mathrsfs}
\usepackage{graphicx}

\usepackage{amsmath,amsthm}

\usepackage{algorithm} %format of the algorithm
\usepackage{algorithmic} %format of the algorithm
\usepackage{multirow} %multirow for format of table
\usepackage{xcolor}
\usepackage{cite}
\usepackage{tabularx}

\usepackage{amsmath}
\usepackage{amssymb}

\usepackage{algorithm}
\usepackage{algorithmic}
\ifCLASSINFOpdf
\else
\fi
\begin{document}
\title{A High-Accuracy Adaptive Beam Training Algorithm for MmWave Communication}

\author{\IEEEauthorblockN{Zihan Tang, Jun Wang, Jintao Wang, and Jian Song}
\IEEEauthorblockA{Beijing National Research Center for Information Science and Technology (BNRist),\\
Tsinghua University, Beijing, China}
Email: tangzh14@mails.tsinghua.edu.cn
%\and
%\IEEEauthorblockN{Homer Simpson}
%\IEEEauthorblockA{Twentieth Century Fox\\
%Springfield, USA\\
%Email: homer@thesimpsons.com}
%\and
%\IEEEauthorblockN{James Kirk\\ and Montgomery Scott}
%\IEEEauthorblockA{Starfleet Academy\\
%San Francisco, California 96678--2391\\
%Telephone: (800) 555--1212\\
%Fax: (888) 555--1212}
}

%\author{
        %Zihan~Tang%,
        %Jun~Wang,~\IEEEmembership{Member,~IEEE,}
        %Jintao~Wang,~\IEEEmembership{Member,~IEEE,}
        %and~Jian~Song,~\IEEEmembership{Fellow,~IEEE}

%\thanks{The authors are with *****}%Tsinghua National Laboratory for Information Science and Technology, and also with the Department of Electronic Engineering, Tsinghua University, Beijing, 100084, China.}
%\thanks{Manuscript received April 19, 2005; revised August 26, 2015.}
%}

% The paper headers
%\markboth{IEEE Communications Letters,~Vol.~**, No.~**, ****~2017}%
%{**** \MakeLowercase{\textit{et al.}}: Power Allocation Schemes for Downlink NOMA Systems With Finite-Alphabet Constraints}

% make the title area
\maketitle

% As a general rule, do not put math, special symbols or citations
% in the abstract or keywords.
\begin{abstract}
In millimeter wave communications, beam training is an effective way to achieve beam alignment. Traditional beam training method allocates training resources equally to each beam in the pre-designed beam training codebook. The performance of this method is far from satisfactory, because different beams have different beamforming gain, and thus some beams are relatively more difficult to be distinguished from the optimal beam than the others. In this paper, we propose a new beam training algorithm which adaptively allocates training resources to each beam. Specifically, the proposed algorithm allocates more training symbols to the beams with relatively higher beamforming gain, while uses less resources to distinguish the beams with relatively lower beamforming gain. Through theoretical analysis and numerical simulations, we show that in practical situations the proposed adaptive algorithm asymptotically outperforms the traditional method in terms of beam training accuracy. Moreover, simulations also show that this relative performance behavior holds in non-asymptotic regime.
\end{abstract}

% Note that keywords are not normally used for peerreview papers.
\begin{IEEEkeywords}
Millimeter wave communication, beam training, misalignment probability
\end{IEEEkeywords}

% For peer review papers, you can put extra information on the cover
% page as needed:
% \ifCLASSOPTIONpeerreview
% \begin{center} \bfseries EDICS Category: 3-BBND \end{center}
% \fi
%
% For peerreview papers, this IEEEtran command inserts a page break and
% creates the second title. It will be ignored for other modes.
\IEEEpeerreviewmaketitle

\section{Introduction}
Millimeter wave (mmWave) communication is one of the most promising technologies of the fifth generation (5G) communication systems, due to the large spectrum resources in the mmWave bands (30-300 GHz)\cite{boccardi.{2014},andrews.{2014}}. Despite its great potential\cite{rappaport.{2013},akdeniz.{2014}}, there are still many problems to be solved before mmWave communications can be realized and deployed in practice. One of the key challenges is that signals in the mmWave bands experience much severer large-scale pathloss compared to signals in lower frequency bands. To overcome this difficulty, large scale antenna arrays with highly directional beamforming technology is used\cite{roh.{2014},alkhateeb.{2014},health.{2016}}. With perfect channel state information, it is easy to achieve beam alignment, which is important to achieve large power gain of the antenna array system\cite{roh.{2014},health.{2016},sohrabi.{2016}}. However, due to the large number of antennas and hardware resource constraints, accurate estimation of the channel matrix is intractable. Another effective approach to realize beam alignment is beam training, in which transmitter and receiver jointly examine beam pairs from pre-designed beam codebooks to find the strongest multi-path component. This approach does not require explicit estimation of the channel matrix, and is particularly appropriate for the sparse mmWave channel\cite{xiao.{2016},seo.{2016},kok.{2017}}.

There are two different beam-training methods considered in the literature. The first one is exhaustive search, which examines all beam pairs in the pre-designed beam codebook and use the strongest beam pair to transmit data. The training overhead of this method is proportional to the size of the search space and thus can be prohibitive when narrow beams are used. To reduce the training overhead, hierarchical beam search method was proposed\cite{alk.{2014}}. In this method, several codebooks with different beam widths are used. The hierarchical beam search first finds the strongest beam pair in a low-level codebook, and then iteratively refines the search using the beams in the next-level codebook within the beam subspace of the wide-beam pair found. Compared with exhaustive search, hierarchical search can effectively reduce search space. However, hierarchical search has an inherent error propagation problem, which means that if the chosen beam pair is incorrect in any given stage, all subsequent searches will be wrong. Furthermore, an early stage with wide beams is more likely to cause error owing to the relatively low beamforming gain, which lead to a higher chance of failing to find the best beam pair at an early stage. Actually, in terms of misalignment probability, researchers have shown the exhaustive search outperforms the hierarchical search, subject to the same training overhead and beam resolution \cite{liu.{2017}}.

In this paper, the beam training problem is studied from a different perspective and a new beam training algorithm is proposed. Unlike traditional exhaustive search, which allocates training resources equally to different beam pairs in the codebook, the new algorithm has the ability to allocate training resources adaptively. Specifically, when using the proposed algorithm, the beam pairs with higher beamforming gain, which are relatively more difficult to be distinguished from the optimal beam pair, are allocated more training resources, while the beam pairs with lower beamforming gain, which are relatively easier to be distinguished from the optimal beam pair, are allocated less training resources. In terms of beam training performance, we give an asymptotic upper bound of misalignment probability for the proposed algorithm. Using this upper bound, we show that in practical situations, the proposed algorithm asymptotically outperforms traditional exhaustive search. Moreover, numerical simulation shows that this relative performance behavior holds in both asymptotic and non-asymptotic regime.

\begin{figure}[h!]
    \centering
    \includegraphics[width=7.6cm]{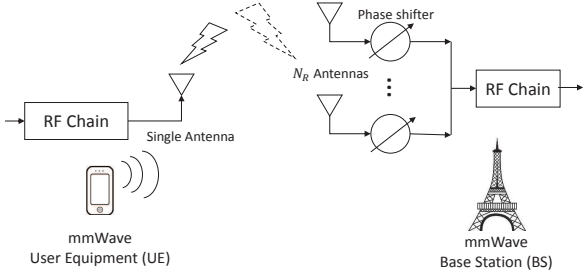}
    \small
    \caption{An illustration of the mmWave uplink communication system considered in this paper.}
    \label{figmodel}
\end{figure}

\section{system description and problem formulation}
We consider mmWave large scale antenna array uplink communication, in which the user equipment (UE) wishes to communicate with the base station (BS) that is equipped with $N_R$ antennas, as depicted in Fig $\ref{figmodel}$. Since the number of antennas at the UE is relatively small and the burden of beam training mainly lies in the BS, for brevity, we assume that the UE is equipped with a single antenna. Similar to \cite{liu.{2017}}, we also assume that the BS has a single RF chain, and thus analog receive beamforming is adopted. In the process of beam training, the UE repeatedly transmits $N$ training symbols, and the BS receives the training symbols using the beams from the pre-designed beam training codebook. Based on the received signal, the BS selects an optimal beam to receive the subsequent data from the UE. We assume that the BS and the UE are synchronized during the whole beam training and data transmission process.

\begin{figure}[h!]
    \centering
    \includegraphics[width=5.2cm]{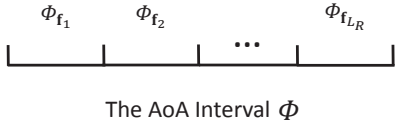}
    \small
    \caption{An illustration of the coverage intervals of each beam in the beam training codebook and the whole AoA interval.}
    \label{figmodel2}
\end{figure}

Denote $\Phi$ as the Angle of Arrival (AoA) search interval of the BS. Let $\mathbf{f}_{l}\in \mathbb{C}^{1\times N_R}$ denote an arbitrary BS receive beamforming vector, with coverage interval $\Phi_{\mathbf{f}_{l}}$. For brevity, $\mathbf{f}_l$ is also used to refer to the corresponding beam in this paper. Let $\mathcal{C}_{R}=\{\mathbf{f}_{l}, l=1,2,\cdots,L_R\}$ denote the pre-designed receive beam training codebook, whose size equals $L_R$. In general, the beams in $\mathcal{C}_{R}$ jointly cover $\Phi$, as depicted in Fig $\ref{figmodel2}$, where the coverage intervals of each beam in the codebook and the whole search interval $\Phi$ are illustrated by line segments. The object of beam training is to identify the best beam in the beam training codebook $\mathcal{C}_{R}$. For example, when traditional exhaustive search is used, the BS allocates $N$ training symbols equally to all $L_R$ beams, which means that each beam of the codebook is used to receive $\frac{N}{L_R}$ consecutive training symbols. After signal processing at the BS, the beam with the strongest beamforming gain will be used to receive data till the next beam training block begins. %Since the hierarchical search has worse performance than exhaustive search in terms of misalignment probability, we do not describe the details of the former here and refer to [].

%In order to characterize the performance of our proposed beam training algorithm and compare it with exhaustive search,
Assume that the channel is frequency-flat and block fading, and remains unchanged during the beam training process. Let $\mathbf{s}\in\mathbb{C}^{1\times K}$ be the $K_l$ identical training symbols received by the BS using the beam corresponding to $\mathbf{f}_{l}$. Then the received signal can be represented as
\begin{equation}
\label{model1}
\begin{aligned}
\mathbf{y}_l &= \mathbf{f}_{l} \mathbf{h}  \mathbf{s} + \mathbf{f}_l \mathbf{Z}_l\\
             &= h_l \mathbf{s} + \mathbf{z}_l,\ \ \  l=1,\cdots,L_R.
\end{aligned}
\end{equation}
where $\mathbf{h}\in \mathbb{C}^{N_R\times 1}$ is the mmWave channel vector, $\mathbf{Z}_l\in \mathbb{C}^{N_R\times K_l}$ is the noise matrix with i.i.d. components $\sim\mathcal{CN}(0,\sigma^2)$, and
\begin{equation}
\label{gain1}
h_l \triangleq \mathbf{f}_{l} \mathbf{h},
\end{equation}
is defined as the effective channel after receive beamforming. Due to the single RF chain constraint, the entries of the receive beamforming vector are either of constant modulus or zero. Moreover, the receive beamforming vector is assumed to satisfy $\Vert\mathbf{f}_l\Vert_2^2=1$. Therefore, the effective noise vector $\mathbf{z}_l\in \mathbb{C}^{1\times K_l}$ has i.i.d. components $\sim\mathcal{CN}(0,\sigma^2)$. We further assume that all training symbols are transmitted with the same power $P_T$, which means $\Vert\mathbf{s}\Vert_2^2=K_lP_T$.

%In order to characterize the performance of beam training, we now formalize the channel model.
Throughout the analysis of this paper, we consider a simple single-path deterministic mmWave channel model which is tractable and catches the essence of the problem. The key insights generated from the analysis of this paper can be applied to the more general cases, in which the channel has multi-path components and the BS is provided with more than one RF chains. %We will show by simulations that the key insights generated from the analysis of this simple single-path channel model continue to apply to scenarios in which there are multi-paths with dominant path.
For the single-path case considered in this paper, the channel vector $\mathbf{h}$ can be represented as
\begin{equation}
\label{channel}
\mathbf{h} = \alpha\mathbf{u}^\dag (\phi ),
\end{equation}
where $\vert\alpha\vert^2$ is the path gain and $\mathbf{u} (\phi )\in \mathbb{C}^{1\times N_R}$ is the steering vector corresponding to AoA $\phi$. For instance, if uniform linear array with half wave-length antenna spacing is equipped at the BS, the steering vector $\mathbf{u} (\phi )$ can be represented as
\begin{equation}
\label{steering}
\mathbf{u}\left(\phi\right) = \left[1,e^{j\pi\text{sin}\left(\phi\right)},\cdots,e^{j\pi\left(N_R-1\right)\text{sin}\left(\phi\right)}\right].
\end{equation}
From $\left(\ref{gain1}\right)$ and $\left(\ref{channel}\right)$, the effective channel gain under any fixed receive beamgorming can be represented as
\begin{equation}
\label{eff_gain}
\begin{aligned}
g_l \triangleq \vert h_l\vert^2 & = \vert \alpha \left(\mathbf{f}_l \mathbf{u}^{\dag}\left(\phi\right)\right)\vert^2,
\end{aligned}
\end{equation}
where $\vert\mathbf{f}_l \mathbf{u}^{\dag}\left(\phi\right)\vert^2$ is the receive beamforming gain of $\mathbf{f}_l$ at AoA $\phi$. The beam training problem is essentially identifying the beam index corresponding to the strongest effective channel gain $\vert h_l\vert^2$, which is
\begin{equation}
\label{opt}
l_{\text{opt}} = \argmax_{l=1,\cdots,L_R} g_l.
\end{equation}

 Without noise, the BS can easily compute the exact value of $h_l$, and then gives the exact solution $l_\text{opt}$. However, since the received signal is swayed by unknown noise, $h_l$ can not be computed exactly. The signal processing method in \cite{liu.{2017}} is adopted here, in which the received signal is match-filtered with the training sequence $\mathbf{s}$ and normalized by $\frac{\sigma^2\Vert\mathbf{s}\Vert^2}{2}$. When each beam is allocated the same number of training symbols (in traditional exhaustive search, $K_l=\frac{N}{L_R}$ for each $l$), the estimated optimal beam index is selected as the one that gives rise to the strongest processing output
\begin{equation}
\label{opt_est}
\hat{l}_{\text{opt}} = \argmax_{l=1,\cdots,L_R} T_l\left(K_l\right),
\end{equation}
in which $T_l\left(K_l\right)$ is defined as the processing output of the $K_l$ training symbols received by the beamforming vector $\mathbf{f}_l$
\begin{equation}
\label{processing}
T_l\left(K_l\right) = \frac{2\vert\mathbf{y}_l\mathbf{s}^{\dag}\vert^2}{\sigma^2\Vert\mathbf{s}\Vert^2}.
\end{equation}

As explained in \cite{liu.{2017}}, $T_l\left(K_l\right)$ essentially captures the energy received by beam $\mathbf{f}_l$, whose expectation is proportional to the number of received training symbols $K_l$, the transmit power $P_T$, as well as the effective channel gain $g_l$. Actually, $T_l\left(K_l\right)$ admits a non-central chi-square distribution with 2 degrees of freedom (DoFs) and a non-centrality parameter
\begin{equation}
\label{lamda}
\lambda_l = \frac{2\vert h_l\vert^2\Vert\mathbf{s}\Vert^2}{\sigma^2}=\frac{2K_lP_Tg_l}{\sigma^2},
\end{equation}
i.e., we have
\begin{equation}
\label{ncx2}
T_l\left(K_l\right) \sim \mathcal{X}_2^2\left(\lambda_l\right),\ \ l=1,2,\cdots,L_R.
\end{equation}
In traditional exhaustive search, since $K_l=N/L_R$ for all $l=1,2,\cdots,L_R$, and all training symbols are transmitted with the same pwoer $P_T$, we can see from ($\ref{lamda}$) that larger $g_l$ gives larger $\lambda_l$ and ($\ref{opt_est}$) is intuitively effective to estimate the optimal beam index. Now it is clear that a misalignment event occurs whenever the estimated optimal beam index $\hat{l}_{\text{opt}}$ is not equal to the true optimal beam index $l_{\text{opt}}$. In the next section, we will propose a new beam training algorithm and analyze the asymptotic performance of it in terms of misalignment probability.

\section{an adaptive beam training algorithm and performance analysis}
\label{adaptive}
In this section, two lemmas relevant to the asymptotic performance of traditional exhaustive search are introduced at first. The intuition of these two lemmas inspires the development of the beam training algorithm in this paper. Then details of the new algorithm as well as an asymptotic performance upper bound is given. Using this bound, we concludes that our proposed algorithm asymptotically outperforms traditional exhaustive search. Before that, we remark that the analysis in this paper focuses on the impact of training resources that tends to be large (infinite). As explained in \cite{liu.{2017}}, this is motivated by the large coherence bandwidth of mmWave band, and thus the number of symbols accommodated within a coherent time interval can be very large. The asymptotic analysis is hence useful to provide guideline on practical system designs. Actually, the numerical simulation in the next section shows that this relative performance behavior holds in both the asymptotic and non-asymptotic regime.

Now we introduce some notations which will be used later. Define the normalized beamforming gain as $\xi_l\triangleq\frac{2P_Tg_l}{\sigma^2}$ ($l=1,\cdots,L_R$). For $l\neq l_{\text{opt}}$, we introduce the following suboptimality measure of beam $l$
\begin{equation}
\label{subo}
\Delta_l=\sqrt{\xi_{l_{\text{opt}}}}-\sqrt{\xi_l}.
\end{equation}
We also define $\Delta_{l^*}$ as the minimal gap
\begin{equation}
\label{decondsub}
\Delta_{l^*}=\min_{l\neq l_{\text{opt}}}\Delta_l,
\end{equation}
which is also the suboptimality measure of the second best beam. For reasons that will be obvious later, we introduce the notation $\left(l\right)\in\left\{1,2,\cdots,L_R\right\}$ to denote the $l$-th best beam (with ties beak arbitrarily), hence
\begin{equation}
\label{permut}
\Delta_{l^*}=\Delta_{\left(1\right)}=\Delta_{\left(2\right)}\leq\Delta_{\left(3\right)}\leq\cdots\leq\Delta_{\left(L_R\right)}.
\end{equation}

In exhaustive search, the training symbols are equally allocated to all the beams in $\mathcal{C}_{R}$, which means that $K_l=\frac{N}{L_R}$ for all $l=1,2,\cdots,L_R$. It is straightforward to see that the misalignment probability of exhaustive search can be represented as
\begin{equation}
\label{mis_ex}
\begin{aligned}
%\text{Pr}\left\{ \hat{l}_{\text{opt}}\neq l_{\text{opt}}\right\}
p_{\text{miss}}^{\text{ex}}\left(N\right) &=\text{Pr}\left\{ \hat{l}_{\text{opt}}\neq l_{\text{opt}}\right\}\\
&= \text{Pr}\left\{ \bigcup_{l=1,l\neq l_{\text{opt}}}^{L_R} T_{l_{\text{opt}}}(\frac{N}{L_R})< T_l(\frac{N}{L_R})\right\}.
\end{aligned}
\end{equation}
Using the union bound, we can derive an upper bound of $p_{\text{miss}}^{\text{ex}}\left(N\right)$, which is
\begin{equation}
\label{up_ex}
\begin{aligned}
p_{\text{up}}^{\text{ex}}\left(N\right) = \sum_{l=1,l\neq l_{\text{opt}}}^{L_R}\text{Pr}\left\{ T_{l_{\text{opt}}}(\frac{N}{L_R})< T_l(\frac{N}{L_R}) \right\}.
\end{aligned}
\end{equation}
Using the large deviation theory, the authors in \cite{liu.{2017}} have proved an important result in terms of the misalignment probability of exhaustive search. Without proof, we restate the result here in the form of the following lemma
\newtheorem{myth}{Lemma}
\begin{myth}%[Asymptotic performance of exhaustive seearch]
\label{asym_ex}
In exhaustive search, assume that the total number of training symbols is $N$, the number of beams is $L_R$, and each beamforming vector is allocated $\frac{N}{L_R}$ training symbols. Then the misalignment probability $p_{\text{miss}}^{\text{ex}}(N)$ as well as the upper bound $p_{\text{up}}^{\text{ex}}(N)$ satisfies
\begin{equation}
\label{performance}
\begin{aligned}
\lim_{N\to \infty}\frac{1}{N}\log p_{\text{miss}}^{\text{ex}}(N) &= \lim_{N\to \infty}\frac{1}{N}\log p_{\text{up}}^{\text{ex}}(N)\\
&=-\frac{1}{4L_R\Delta_{l^*}^{-2}},
\end{aligned}
\end{equation}
where $l^* \triangleq \argmax_{l=1,\cdots,L_R,\ l\neq l_{\text{opt}}}\xi_l$.
\end{myth}

As lemma $\ref{asym_ex}$ shows, the asymptotic performance of traditional exhaustive search is proportional to the difference between square roots of the strongest beamforming gain and the second strongest beamforming gain. It seems that the performance mainly depends on how well the algorithm can distinguish the optimal beam from the second-best one. In fact, the algorithm needs to distinguish the optimal beam from all the other beams in $\mathcal{C}_{R}$, and exhaustive search implements this task by allocating the same number of training symbols to all the beams. Here we want to break the constraint of uniform allocation of training resource to increase the performance of beam training accuracy. In order to demonstrate the essence of the problem more clearly, we need to analyze the pair-wise misalignment probability. The result is summarized in the following lemma $\ref{theo_pair}$, the proof of which can be found in \cite{liu.{2017}}.

\newtheorem{myth2}{Lemma}
\begin{myth}%[Asymptotic pair-wise misalignment probability]
\label{theo_pair}
Assume that the optimal beam $\mathbf{f}_{l_{\text{opt}}}$ and the $\tilde{l}$-th-best beam $\mathbf{f}_{\left(\tilde{l}\right)}$ ($\tilde{l}=2,3,\cdots,L_R$) are allocated the same number of training symbols: $K_{\left(\tilde{l}\right)}=K_{l_{\text{opt}}}=K$. Then the pair-wise misalignment probability satisfies
\begin{equation}
\label{pair}
\lim_{K\to \infty}\frac{1}{K}\log\text{Pr}\left\{ T_{l_{\text{opt}}}\left(K\right)< T_{\tilde{l}}\left(K\right) \right\}=-\frac{1}{4\Delta_{\tilde{l}}^{-2}}.
\end{equation}
\end{myth}

Lemma $\ref{theo_pair}$ shows that the number of training symbols needed to achieve a fixed level of pair-wise misalignment probability of distinguishing the optimal beam from the $\tilde{l}$-th-best ($\tilde{l}=2,\cdots,L_R$) beam is asymptotically proportional to $\Delta_{\left(\tilde{l}\right)}^{-2}$. It means that the beam with a larger suboptimality measure is relatively easier to be distinguished from the optimal beam while the beam with a smaller suboptimality measure is harder to be distinguished. In practical beam training codebook design, the beamforming gain of all the beams in the codebook usually distributes like the absolute values of samples of a $Sinc$ function\cite{gao.{2017}} or other irregular functions \cite{alk.{2014},xiao.{2016}}. It means that there are only a few small elements in $\left\{\Delta_{\left(2\right)},\Delta_{\left(3\right)},\cdots,\Delta_{\left(L_R\right)}\right\}$, and the others are relatively large. The inspiration is that traditional exhaustive search wastes a lot of training resources on the beams that are easy to be distinguished. Intuitively, if more training symbols are allocated to the beams with smaller suboptimality measure, the performance of training accuracy can be improved. However, the main problem is that the suboptimality measure of each beam is unknown, and thus the optimal allocation of training resources can not be accomplished beforehand. Inspired by the idea of successive rejects algorithm in the best arm identification problem \cite{audi.{2016}}, here we propose an adaptive beam training algorithm that can progressively reject the beams which seem to be suboptimal.

The details of the adaptive beam training algorithm is given in Algorithm 1. Informally it proceeds as follows. First the algorithm divides the training resources (i.e. the N training symbols) into $L_R-1$ phases. At the end of each phase, the algorithm discards the beam with the lowest processing output. During the next phase, it allocates equal number of training symbols to each beam which has not been discarded yet. At the end of the last phase, the estimated optimal beam is selected as the last surviving beam. The length (number of training symbols) of each phase is carefully chosen to obtain a good performance. More precisely, the first discarded beam is allocated $n_1=\lceil\frac{1}{\overline{\text{log}}\left(L_R\right)}\frac{N-L_R}{L_R}\rceil$ training symbols, the second discarded beam is allocated $n_2=\lceil\frac{1}{\overline{\text{log}}\left(L_R\right)}\frac{N-L_R}{L_R-1}\rceil$ training symbols,..., and the last two discarded beams are allocated $n_{L_R-1}=\lceil\frac{1}{\overline{\text{log}}\left(L_R\right)}\frac{N-L_R}{2}\rceil$ training symbols. The proposed algorithm does not exceed the budget of $N$ training symbols, since we have
\begin{equation}
\begin{aligned}
\label{num_sum}
&\ \ \ \ n_1+n_2+\cdots+n_{L_R-1}+n_{L_R-1}\\
&\leq L_R+\frac{N-L_R}{\overline{\text{log}}\left(L_R\right)}\left(\frac{1}{2}+\sum_{l=1}^{L_R-1}\frac{1}{L_R+1-l}\right)\\
&\leq N,
\end{aligned}
\end{equation}
in which the first inequality is derived from the fact $\lceil x\rceil\leq x+1$ for any real value $x$, and the second inequality is derived from the definition of $\overline{\text{log}}\left(L_R\right)$. For $L_R=2$, up to rounding effects, algorithm 1 is just traditional exhaustive search.

\begin{algorithm}[t]
%\begin{small}
\caption{An adaptive beam training algorithm}
%\textbf{Input:} The beamforming codebook $\mathcal{C}_{R}$, training overhead $N$, transmitting power $P_T$.\\
\textbf{Initialize:} Let $\mathcal{A}_1=\left\{1,2,\cdots,L_R\right\}$, and
\begin{equation*}
\overline{\text{log}}\left(L_R\right)=\frac{1}{2}+\sum_{i=2}^{L_R}\frac{1}{i}.
\end{equation*}
Let $n_0=0$ and for $k\in\left\{1,\cdots,L_R-1\right\}$,
\begin{equation*}
n_k=\lceil\frac{1}{\overline{\text{log}}\left(L_R\right)}\frac{N-L_R}{L_R+1-k}\rceil,
\end{equation*}
in which $\lceil\cdot\rceil$ is the ceiling function.\\
\textbf{for} $k=1,2,\cdots,L_R-1$ \textbf{do}
\begin{enumerate}
  \item For each $l\in\mathcal{A}_k$, use the beamforming vector $\mathbf{f}_l$ to receive the consecutive $n_k-n_{k-1}$ training symbols.
  \item Let $\mathcal{A}_{k+1}=\mathcal{A}_{k}\backslash \argmin_{l\in\mathcal{A}_k}T_l\left(n_k\right)$ (if there is a tie, break it randomly).
\end{enumerate}
\textbf{end}\\
\textbf{Output:} The index of the estimated optimal beam, which is denoted as $\tilde{l}_{opt}$, is selected to be the unique element of $\mathcal{A}_{L_R}$
%\end{small}
\end{algorithm}

It is straightforward to see that as algorithm 1 proceeds, the beams with relatively larger suboptimality measure tend to be discarded at the earlier stages, and thus are allocated less training resources. This makes the proposed algorithm more efficient to use the training resources, and thus have better performance of training accuracy. In the following proposition, we give an asymptotic upper bound of the misalignment proabbility of algorithm 1, which is defined as $p_{\text{miss}}^{\text{adpt}}(N)=\text{Pr}\left\{ \tilde{l}_{\text{opt}}\neq l_{\text{opt}}\right\}$.

\newtheorem{mypro}{Proposition}
\begin{mypro}%[Asymptotic performance of exhaustive seearch]
\label{asym_upper}
In algorithm 1, assume that the number of training symbols is $N$, then the misalignment probability $p_{\text{miss}}^{\text{adpt}}(N)$ satisfies
\begin{equation}
\label{peroformance_up}
\lim_{N\to \infty}\frac{1}{N}\log p_{\text{miss}}^{\text{adpt}}(N) \leq -\frac{1}{4\overline{\text{log}}\left(L_R\right)H},
\end{equation}
where $H=\max_{l\in\left\{1,2,\cdots,L_R\right\}} l\Delta_{\left(l\right)}^{-2}$.
\end{mypro}
\begin{proof}[\textbf{Proof}]%[Asymptotic performance of exhaustive seearch]
%Similar to ($\ref{processing}$), we denote $T_l\left(K\right)$ as the processing output when $K$ pilot symbols is received by the beamforming vector $\mathbf{f}_l$, just to emphasize $K$.
According to Lemma 2, the pair-wise misalignment probability can be represented as
\begin{equation}
\footnotesize
\begin{aligned}
\label{proof_pair}
&\ \ \ \text{Pr}\left\{T_{l_{\text{opt}}}\left(n_i\right)\leq T_{\left(l\right)}\left(n_i\right)\right\}=\exp\left(-\frac{n_i\Delta_{\left(l\right)}^2}{4}+o_{\left(l\right)}\left(n_i\right)\right),
\end{aligned}
\end{equation}
in which $o_{\left(l\right)}\left(n_i\right)$ is a high order infinitesimal of $n_i$ that satisfies $\lim_{n_i\to \infty}\frac{o_{\left(l\right)}\left(n_i\right)}{n_i}=0$. We can assume that training symbols are received before the beginning of beam training. Thus $T_l\left(K\right)$ is well defined even if $\mathbf{f}_l$ has not been used to receive $K$ training symbols. During phase $i$, at least one of $i$ worst beams is surviving. So, if the optimal beam $l_{\text{opt}}$ is discarded at the end of phase $i$, it means that $T_{l_{\text{opt}}}\left(n_i\right)\leq\text{max}_{l\in\left\{\left(L_R\right),\left(L_R-1\right),\cdots,\left(L_R+1-i\right)\right\}}T_{l}\left(n_i\right)$. Based on (\ref{proof_pair}) and the union bound, the misalignment probability of algorithm 1 satisfies
\begin{equation}
\footnotesize
\begin{aligned}
\label{proof}
p_{\text{miss}}^{\text{adpt}}(N)
&\leq\sum_{i=1}^{L_R-1}\text{Pr}\left\{T_{l_{\text{opt}}}\left(n_i\right)\leq\max_{l\in\left\{\left(L_R\right),\cdots,\left(L_R+1-i\right)\right\}}T_{l}\left(n_i\right)\right\}\\
%&\leq \sum_{i=1}^{L_R-1}\text{Pr}\left\{\bigcup_{l=L_R+1-i}^{L_R}T_{l_{\text{opt}}}\left(n_i\right)\leq T_{\left(l\right)}\left(n_i\right)\right\}\\
&\leq \sum_{i=1}^{L_R-1}\sum_{l=L_R+1-i}^{L_R-1}\text{Pr}\left\{T_{l_{\text{opt}}}\left(n_i\right)\leq T_{\left(l\right)}\left(n_i\right)\right\}\\
%&\leq \sum_{i=1}^{L_R-1}\sum_{l=L_R+1-i}^{L_R-1} \exp\left(-\frac{n_i\Delta_{\left(l\right)}^2}{4}+o_{\left(l\right)}\left(n_i\right)\right)\\
&\leq \sum_{i=1}^{L_R-1} \sum_{l=L_R+1-i}^{L_R-1}\exp\left(-\frac{n_i\Delta_{\left(L_R+1-i\right)}^2}{4}+o_{\left(l\right)}\left(n_i\right)\right)\\
&\leq \sum_{i=1}^{L_R-1} \sum_{l=L_R+1-i}^{L_R-1}\exp\left(-\frac{N-L_R}{4\overline{\text{log}}\left(L_R\right)H}+o_{\left(l\right)}\left(n_i\right)\right),
\end{aligned}
\end{equation}
in which the last inequality holds because
\begin{equation}
\begin{aligned}
\label{proof_con2}
n_i\Delta_{\left(L_R+1-i\right)}^2 &\geq \frac{N-L_R}{\overline{\text{log}}\left(L_R\right)}\frac{1}{\left(L_R+1-i\right)\Delta_{\left(L_R+1-i\right)}^{-2}}\\
&\geq \frac{N-L_R}{\overline{\text{log}}\left(L_R\right)H}.
\end{aligned}
\end{equation}
Using definition of $n_i$ and $o_{l}\left(n_i\right)$, we have $\lim_{N\to \infty}\frac{o_{\left(l\right)}\left(n_i\right)}{N} = 0$. This means that each term in the summation ($\ref{proof}$) decays exponentially at rate $4\overline{\text{log}}\left(L_R\right)H$ as $N$ goes to infinity. Since there are only finite terms in the summation ($\ref{proof}$), we conclude that the summation decays exponentially at the same rate, which completes the proof.
\end{proof}

Now we compare the performance of algorithm 1 and traditional exhaustive search. Define the index in the maximization $H$ as $l_H=\argmax_{l\in\left\{1,2,\cdots,L_R\right\}} l\Delta_{\left(l\right)}^{-2}$.
As mentioned before, in practical beamforming codebook design, only the first few elements in $\left\{\Delta_{\left(2\right)},\Delta_{\left(3\right)},\cdots,\Delta_{\left(L_R\right)}\right\}$ are small, and the rest are relatively large. It implies that $l_H$ is usually much smaller than $L_R$, especially when $L_R$ is very large, which corresponds to the situation when the high-resolution beam training codebook is used. Thus we have:
\begin{equation}
H=l_H\Delta_{\left(l_H\right)}^{-2} \leq l_H\Delta_{l*}^{-2} \ll L_R\Delta_{l^*}^{-2}.
\end{equation}
which further implies that $\overline{\text{log}}\left(L_R\right)H$ is smaller than $L_R\Delta_{l^*}^{-2}$. Using the results in proposition 1 and lemma 1, we have
\begin{equation}
\label{conclude}
\lim_{N\to \infty}\frac{1}{N}\log p_{\text{miss}}^{\text{adpt}}(N) < \lim_{N\to \infty}\frac{1}{N}\log p_{\text{miss}}^{\text{ex}}(N),
\end{equation}
It means that the proposed adaptive algorithm asymptotically outperforms traditional exhaustive search in terms of training accuracy. Here we remark that the above analysis is not a formal mathematical proof, but it is true in almost all practical situations. In the next section, numerical simulation shows that this relative performance behavior holds in both asymptotic and non-asymptotic regime.

\section{numerical simulation}
In this section, we compare the performance between the proposed algorithm and traditional exhaustive search through numerical examples. In the following simulations, we assume that the AoA search interval $\Phi$ covers the angular space $\left[-\frac{\pi}{2},\frac{\pi}{2}\right]$. We choose the simplest codebook design \cite{seo.{2016},gao.{2017}}, in which the $l$-th beamforming vector is
\begin{equation}
\label{simple_beam}
\textbf{f}_l=\frac{1}{\sqrt{L_R}}\left[1,e^{-j2\pi\theta_l},\cdots,e^{-j2\pi\left(L_R-1\right)\theta_l}\right],
\end{equation}
and $\theta_l=-\frac{1}{2}+\frac{l-1}{L_R}$, $l=1,2,\cdots, L_R$. We choose the number of receive antennas $L_R=64$. As for the channel, we choose the path AoA $\phi$ randomly from the AoA interval $\left[-\frac{\pi}{2},\frac{\pi}{2}\right]$. For example, we choose $\phi=0.47$ in the following simulation. The effective beamforming gain $g_l$ ($l=1,2,\cdots,L_R$) is given in the left part of Fig. $\ref{process}$. It is straightforward to see that there are only a few beamforming gain which are large and the others are relatively very small, as analyzed in the previous section. Actually, if we let $\phi$ be a uniform random variable in $\left[-\frac{\pi}{2},\frac{\pi}{2}\right]$, the conclusion holds with a high probability. We remark here that even if more complex codebook design is used \cite{alk.{2014}}, this conclusion also holds true. More detailed analysis and simulations will be given in our future work.
\begin{figure}[h!]
    \centering
    \includegraphics[width=7.2cm]{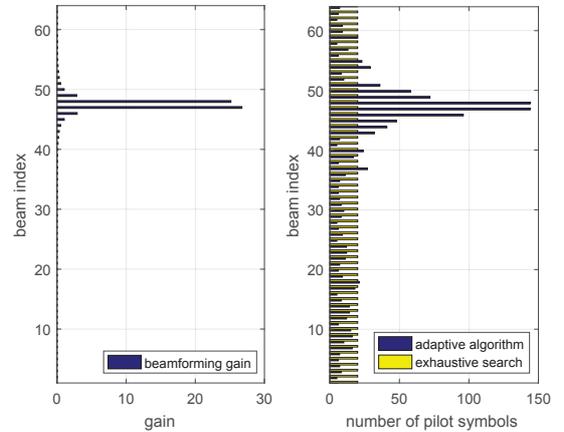}
    \small
    \caption{Left: horizontal coordinates is the effective beamforming gain $g_l$, vertical coordinates is the index of beams; Right: horizontal coordinates is the number of training symbols allocated to each beam, vertical coordinates is the index of beams. $N=1280$.}
    \label{process}
\end{figure}

\begin{figure}[h!]
    \centering
    \includegraphics[width=7.2cm]{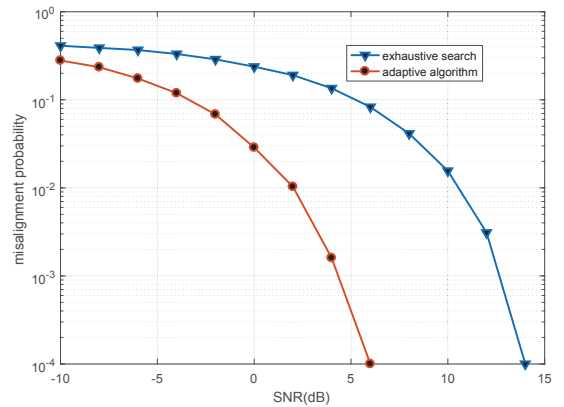}
    \small
    \caption{Misaligment probability versus SNR. The blue line corresponds to traditional exhaustive search, and the red line corresponds to the adaptive algorithm. $N=1280$. }
    \label{snr}
\end{figure}

In order to compare the performance of the proposed adaptive algorithm with traditional exhaustive search, we run the two algorithms and show the number of training symbols allocated to each beam in the right part of Fig. $\ref{process}$. The total number of training symbols is set to be $N=1280$. It is obvious that the proposed algorithm allocates more training symbols to the beams with larger beamforming gain, just as analyzed before. Furthermore, we set the channel path gain as $\alpha=1$, and define the signal-to-noise-ratio (SNR) as $\text{SNR}=\frac{P_T}{\sigma^2}$. The misalignment probability of the two algorithm versus SNR is given in Fig. $\ref{snr}$, which shows that the proposed adaptive algorithm has much lower misalignment probability than traditional exhaustive search at all considered SNR values.

\begin{figure}[h!]
    \centering
    \includegraphics[width=7.4cm]{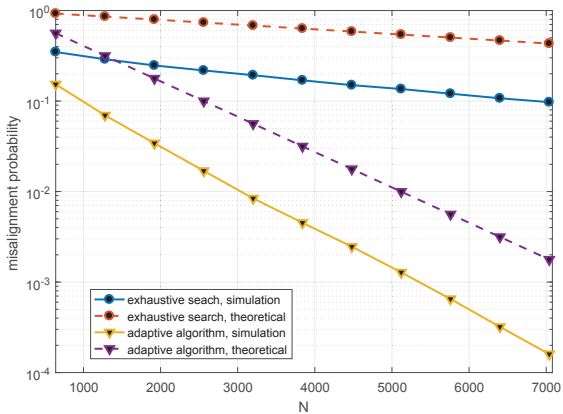}
    \small
    \caption{Misaligment probability versus the total number of training symbols $N$, $\text{SNR}=-2\text{dB}$}
    \label{num}
\end{figure}

In order to further compare the performance of the two algorithms and verify the theoretical results in the previous section, we fix $\text{SNR}$ and study the performance of misalignment probability in terms of the total number of training symbols $N$. Taking $\text{SNR}=-2\text{dB}$ as an example, we validate our analysis in Fig. $\ref{num}$. It is straightforward to see that the blue solid line is almost parallel to the red dashed line, which verifies the asymptotic results ($\ref{performance}$). What is surprising is that the yellow solid line is almost parallel to the purple dashed line, which indicates that the upper bound ($\ref{peroformance_up}$) is true and the bound may be tight. %So, even if we can not give a rigorous proof, we guess that the following result is true:
%\begin{equation}
%\label{guess}
%\lim_{N\to \infty}\frac{1}{N}\log p_{\text{miss}}^{\text{adpt}}(N) = -\frac{1}{4\overline{\text{log}}\left(L_R\right)H}.
%\end{equation}
At last, compare the blue solid line with the yellow solid line in Fig. $\ref{num}$, we can see that with the chosen parameter, the proposed adaptive algorithm asymptotically outperforms traditional exhaustive search ($\lim_{N\to \infty}\frac{1}{N}\log p_{\text{miss}}^{\text{adpt}}(N)\approx -9.0\times 10^{-4}$, $\lim_{N\to \infty}\frac{1}{N}\log p_{\text{miss}}^{\text{ex}}(N)\approx -1.2\times 10^{-4}$), and this relative performance behaviour also holds when $N$ is small.

\section{conclusions}
In millimeter wave communications, beam training is an effective way to achieve beam alignment. In this paper, the beam training problem is studied from a different perspective and a new beam training algorithm is proposed. Unlike traditional exhaustive search, which allocates training resources uniformly to different beams in the pre-designed codebook, the new algorithm adaptively allocates training resources. Specifically, it allocates more training resources to the beams with higher beamforming gain, which are relatively more difficult to be distinguished from the best beam, while uses less resources to distinguish the beams with low beamforming gain, which are relatively easier to identify. We give an asymptotic upper bound of misalignment probability of the proposed algorithm. Using this upper bound and numerical simulations, we show that in practical situations, the proposed adaptive algorithm asymptotically outperforms the traditional exhaustive search in terms of misalignment probability, subject to the same training overhead. Moreover, numerical simulation shows that this relative performance behavior also holds in the non-asymptotic regime. In the future, we will give more detailed analysis of the proposed algorithm in this paper and apply it to more complex situations.

% if have a single appendix:
%\appendix[Proof of the Zonklar Equations]
% or
%\appendix  % for no appendix heading
% do not use \section anymore after \appendix, only \section*
% is possibly needed

% use appendices with more than one appendix
% then use \section to start each appendix
% you must declare a \section before using any
% \subsection or using \label (\appendices by itself
% starts a section numbered zero.)
%

\section*{Acknowledgment}
This work was supported by the Natural Science Foundation of Guangdong Province (Grant No. 2015A030312006)

% you can choose not to have a title for an appendix
% if you want by leaving the argument blank

% use section* for acknowledgment
%\section*{Acknowledgment}

%The authors would like to thank...

% Can use something like this to put references on a page
% by themselves when using endfloat and the captionsoff option.
\ifCLASSOPTIONcaptionsoff
  \newpage
\fi

\end{document}